\documentclass[a4paper,10pt]{article}
\usepackage{QED}
\usepackage{amsmath}
\usepackage{amssymb}
\usepackage{a4}

\def\b{\beta}
\def\g{\gamma}

\def\G{\Gamma}

\def\t{\tau}

\def\d{\delta}

\def\l{\lambda}

\def\s{\sigma}

\def\f{\rightarrow}
\def\tr{\triangleright}

\def\v{\vdash}

\def\et{\wedge}
\def\<{\langle}
\def\>{\rangle}
\def\F{\displaystyle\frac}
\newtheorem{theorem}{Theorem}[section]
\newtheorem{lemma}{Lemma}[section]

\newtheorem{corollary}{Corollary}[section]

\newtheorem{definition}{Definition}[section]
\newtheorem{notation}{Notation}[section]

\begin{document}

\title{A short proof that adding some permutation rules to $\b$ preserves $SN$}

\author{Ren\'e David \\
LAMA - Equipe LIMD -
 Universit\'e de Chamb\'ery\\
e-mail : rene.david@univ-savoie.fr }

%\date{}

\maketitle

\begin{abstract}
I show that, if a term is $SN$ for $\beta$,  it remains $SN$ when
some permutation rules  are added.
\end{abstract}

\section{Introduction}

Strong normalization  (abbreviated as $SN$) is a property of
rewriting systems that is often desired. Since about 10 years many
researchers have considered the following question : If a
$\l$-term is $SN$ for the $\b$-reduction, does it remain $SN$ if
some other reduction rules are added ?  They are mainly interested
with permutation rules they introduce to be able to delay some
$\b$-reductions in, for example, {\it let $x$ = ... in ...}
constructions or in {\it calculi with explicit substitutions}.
Here are some papers considering such permutations rules: L.
Regnier \cite{regnier}, F. Kamareddine \cite{fairouz}, E. Moggi
\cite{moggi}, R. Dyckhoff and S. Lengrand \cite{lengrand}, A. J.
Kfoury  and J. B. Wells \cite{wells}, Y. Ohta and M. Hasegawa
\cite{ohta}, J. Esp\'irito Santo \cite{jose-1},  \cite{jose-2},
and \cite{jose-3}.

Some of these papers show that $SN$ is preserved by the addition
of the permutation rules they introduce but, most often, authors
do not consider  the whole set of rules or add   restrictions to
some rules. For example the rule  $( M \ (\lambda x. N \ P ) )
\triangleright
  (\lambda x. (M \ N) \ P )$ is often restricted to the case when
  $M$ is an abstraction (in this case it is usually called
  ${assoc}$).

  I give here a simple and short proof that the permutations rules
  preserve $SN$ when they are added all together and with no
  restriction. It is done as follows. I show
  that every term which is typable in the system (often called system $\cal{D}$) of
  types built with $\rightarrow$ and $\et$ is strongly normalizing for all the rules
($\beta$ and the permutation rules). Since it is well known that a
term is $SN$ for the $\beta$-rule iff it is typable in this
system, the result follows. The proof is an extension of my proof
of $SN$ for the simply typed $\l$-calculus where the main result
is  a substitution theorem (here Theorem \ref{main}):  if $t$ and
$a$ are in $SN$, then so is $t[x:=a]$.

To my knowledge, only one other paper (\cite{jose-2} and
  its recent version \cite{jose-3}) considers all the rules with
  no restriction. The technic used there is completely different
  from the one used in this paper.

\section{Definitions and notations}

\begin{definition}\label{def}
\begin{itemize}
  \item The set of $\l$-terms is defined by the following grammar

$${\cal M} := x\ | \ \l x.  {\cal M} \ | \ ({\cal M} \; {\cal M})$$

  \item  The set ${\cal T}$ of types is defined (simultaneously with the set ${\cal S}$ of simple types) by the following
  grammars where  ${\cal A}$ is a set of atomic constants

$${\cal S} ::= \; {\cal A} \; \mid {\cal T} \f {\cal S}$$
$${\cal T} ::= \; {\cal S} \; \mid {\cal S} \et {\cal T}$$

\item The typing rules are the following where $\G$ is a set of declarations as $x:A$
where $x$ is a variable and the mentioned types ($A, B$) are in
$\cal{T}$:

\begin{center}

$\F{}{\G , x : A \v x : A} $

\vspace{.5cm}

$\F{\G \v M : A \f B \quad \G \v N : A} {\G \v (M \; N) : B }$
\hspace{0.5cm} $\F{\G, x: A \v M : B} {\G \v \l x.M : A \f B}$

\vspace{.5cm}

$\F{\G \v M : A \et B} {\G \v M : A }$ \hspace{0.5cm} $\F{\G \v M
: A \et B} {\G \v M : B }$

\vspace{.5cm}

$\F{\G \v M : A  \quad \G \v M : B} {\G \v M : A \et B }$

\end{center}
\end{itemize}
\end{definition}

\noindent {\bf Remarks and Notation}

\begin{enumerate}
  \item To avoid too many brackets in the lambda terms I will adopt the following
  conventions.
 An application (or a sequence of applications) is always surrounded by brackets (i.e. the application of $M$ to $N$
 is written $(M \ N )$ with a blank between $M$ and $N$) and, as usual,
   application associates to the left i.e. $(M \ N \
  P)$
  means $((M \ N) \ P)$.
An abstraction is always written as $\lambda x.M$ (i.e. there is a
dot after the variable but no blank between the dot and $M$) where
either $M$ is a letter or an application (and thus between
brackets) or another abstraction.

For example $\lambda y.(M N)$ represents  an abstraction and
$(\lambda y.M N)$ a redex.

  \item Note that in the usual definition of the types with intersection
$\rightarrow$ and $\et$ can be used with no restriction. Here we
forbid to have an $\et$ at the right of an $\rightarrow$. For
example $A \f (B \et C)$ is forbidden and must be replaced by $(A
\f B) \et (A \f C)$. It is well known that both systems are
equivalent since it is easily proved that any type derivation in
the unrestricted system can be transformed into a type derivation
in the restricted one. Actually note that, in fact, the type
derivation given by Theorem \ref{D} already satisfies this
restriction.

We have used this restricted version to make simpler the analysis
of type derivations in the proof of Theorem \ref{main}
  \item Also note (this is well known and easy to prove) that any type
derivation can be transformed into a normal derivation i.e. a
derivation in which the introduction of an $\et$ is never
immediately followed by its elimination.
\item The lemmas and theorems using types will be indicated by the mention \textquotedblleft typed".
If a type derivation is given to  $M$,  $type(M)$ will denote the
size (i.e the number of symbols) of the  type of $M$.
\end{enumerate}

\begin{definition}
The reduction rules are the following.
\begin{itemize}
  \item $\beta$ : $(\lambda x. M \ N) \triangleright M[x:=N]$
  \item $\d$ : $(\l y.\l x. M \ N) \triangleright \l x.(\l y. M \
  N)$
\item $\g$ : $(\lambda x. M \ N \ P) \triangleright (\lambda x. (M \ P) \
N)$
   \item ${assoc}$  : $( M \ (\lambda x. N \ P ) ) \triangleright
  (\lambda x. (M \ N) \ P )$
\end{itemize}
\end{definition}

 Using
Barendregt's convention for the names of variables, we  assume
that, in $\g$ (resp. $\d$, $assoc$), $x$ is not free in $P$ (resp.
in $N$,  in $M$).

The rules $\d$ and $\g$ have been introduced by Regnier in
\cite{regnier} and are called there the $\s$-reduction. It seems
that the first formulation of {\em assoc} appears in Moggi
\cite{moggi} in the restricted case where  $M$ is an abstraction
and in a  ``{\em let ... in ...}'' formulation.

Note that $\g$ (resp. $\d$, $assoc$) are called $\theta_1$ (resp.
$ \g$, $\theta_3$) in \cite{wells} and $\pi_1$ or $\s_1$ (resp.
$\s_2$, $\pi_2$) in \cite{jose-3}.

\begin{notation}
\begin{itemize}

  \item If $t$ is a term, $size(t)$ denotes its size.
  \item If $t \in SN$ (i.e. every sequence of reductions starting from $t$
is finite), $\eta(t)$  denotes
   the length of the
longest reduction of $t$. Since various notions of reductions are
considered in this paper, by default these concepts  are relative
to the union of all four reduction rules. When this is not the
case (e.g. $SN$ wrt to $\beta$), then the reduction rule intended
is indicated explicitly.

\item Let $\s$ be a substitution. We say that $\s$ is fair if the  $\s(x)$ for $x
\in dom(\s)$ all have the same type (that will be denoted as
$type(\s)$). We say that $\s \in SN$  if, for each $x \in
dom(\s)$, $\s(x) \in SN$.
\item Let $\s \in SN$ be a substitution and $t$ be a term. We denote by $size(\s,t)$ (resp. $\eta(\s,t)$) the
sum, over $x \in dom(\s)$, of $nb(t,x). size(\s(x))$ (resp.
$nb(t,x). \eta(\s(x))$)
 where $nb(t,x)$ is the number of free occurrences of $x$
in $t$.
\item If $\overrightarrow{M}$ is a sequence of terms, $lg(\overrightarrow{M})$ denotes its length,
$M(i)$  denotes the $i$-th element of the sequence and
$tail(\overrightarrow{M})$  denotes $\overrightarrow{M}$ from
which the first element has been deleted.
 \item Assume $t=(H \  \overrightarrow{M})$ where $H$ is an abstraction or a variable and $lg(\overrightarrow{M})\geq
 1$.

\begin{itemize}

\item If $H$ is an abstraction (in this case we say that $t$ is $\b$-head reducible), then $M(1)$ will
be denoted as $Arg[t]$ and
 $(R' \ tail(\overrightarrow{M}))$ will
be denoted by $B[t]$ where $R'$ is the reduct of the $\beta$-redex
$(H  \ Arg[t])$.
\item  If $H=\l x. N$ and $lg( \overrightarrow{M}) \geq 2$ (in this case we say that
$t$ is $\g$-head reducible), then $(\l x. (N \ M(2))$  $ M(1) \
M(3) \ ... \ M(lg(\overrightarrow{M})))$ will be denoted by
$C[t]$.
\item If $H=\l x. \l y. N$  (in this case we say that
$t$ is $\d$-head reducible), then $(\l y. (\l x. N \ M(1)) \ M(2)
\ ... \ M(lg(\overrightarrow{M})))$ will be denoted by $D[t]$.
\item If $M(i)=(\lambda x. N \ P)$, then the term
$(\lambda x. (H \ M(1) \ ... \ M(i-1) \ N ) \ P \ M(i+1) \ ... \
M(lg(\overrightarrow{M})))$ will be denoted by $A[t, i]$ and we
say that $M(i)$ is the $\b$-redex put in head position.

\end{itemize}
\item Finally, in a proof by induction, IH will denote the induction
hypothesis.

\end{itemize}

\end{notation}

\section{The theorem}

\begin{theorem}
Let $t$ be a term. Assume $t$ is strongly normalizing for $\b$.
Then $t$ is strongly normalizing for $\b$, $\d$, $\g$ and $assoc$.
\end{theorem}
\begin{proof}
This follows immediately from Theorem \ref{D} and corollary
\ref{cor} below.
\end{proof}

\begin{theorem}\label{D}
A term is $SN$ for the $\b$-rule iff it is typable in  system
$\cal{D}$.
\end{theorem}
\begin{proof}
This is a classical result. For the sake of completeness I recall
here the proof of the only if direction given in  \cite{moi}. Note
that it is the only direction that is used in this paper and that
corollary \ref{cor} below actually gives the other direction. The
proof is by induction on $\langle\eta(t),size(t)\rangle$.

-  If $t=\lambda x\;u.$ This follows immediately from the IH.

- If $t=(x\;v_{1}\;...\;v_{n})$.
 By the IH, for every $j$, let $x:A_{j},\Gamma _{j}\vdash v_{j}:B_{j}$. Then $x:\bigwedge
 A_{j}\et
(B_{1},...,B_{n}\rightarrow C),$ $\bigwedge \Gamma _{j}\vdash t:C$
where $C$ is any type, for example any atomic type.

-  If $t=(\lambda x.a \;b\;\overrightarrow{c})$. By the IH,
$(a[x:=b]\;\overrightarrow{c})$ is typable.  If $x$ occurs in $a$,
let $A_{1}\;...\;A_{n}$ be the types of the occurrences of $b$ in
the typing of $(a[x:=b]\;\overrightarrow{c})$. Then $ t $ is
typable by giving to $x$ and $b$ the type $A_{1}\;\et ...\;\et
A_{n}$. Otherwise, by the induction hypothesis $b$ is typable of
type $B$ and then $t$ is typable by giving to $x$ the type $B$.
\end{proof}

{\it From now on, $\triangleright$ denotes the reduction by one of
the rules $\b$, $\d$, $\g$ and $assoc$}.

\begin{lemma}\label{prepa}
\begin{enumerate}
\item The system satisfies subject reduction i.e. if $\G \v t : A$ and $t \triangleright t'$ then
$\G \v t' : A$.

\item If $t \triangleright t'$ then $t[x:=u] \triangleright t'[x:=u]$.
  \item If   $t'=t[x:=u] \in SN$ then  $t \in SN$
  and $\eta(t) \leq \eta(t')$.
\end{enumerate}
\end{lemma}
\begin{proof}
Immediate.
\end{proof}

\begin{lemma}\label{cs-sn}
Let $t=(H \  \overrightarrow{M})$ be such that $H$ is an
abstraction or a variable and $lg(\overrightarrow{M})\geq 1$.
Assume $H, \overrightarrow{M} \in SN$ and that
\begin{enumerate}
\item If $t$ is $\d$-head reducible (resp. $\g$-head reducible, $\b$-head reducible), then $D[t] \in SN$
(resp. $C[t] \in SN$, $Arg[t], B[t] \in SN$).

\item For each $i$ such that $M(i)$ is a $\b$-redex, $A[t,i] \in SN$,

\end{enumerate}
Then $t \in SN$.
\end{lemma}
\begin{proof}
By induction on $\eta(H) + \sum \eta(M(i))$. Show that each reduct
of $t$ is in $SN$. Note that the assumption $H,\overrightarrow{M}
\in SN$ is implied by the others if at least one of them is not
\textquotedblleft empty" i.e. if $t$ is head reducible for at
least one rule.
\end{proof}

\begin{lemma}[typed]\label{prepa2}
If $(t \ \overrightarrow{u}) \in SN$ then $(\lambda x. t \ x \
\overrightarrow{u}) \in  SN$.
\end{lemma}
\begin{proof}
Note that, if $(\lambda x. t \ x \ \overrightarrow{u})$ has a head
redex for the $\d$-rule, its reduct has not the desirable shape
and an induction hypothesis will not be applicable. We thus
generalize a bit the statement with the notion of left context,
i.e. a context with exactly one hole on the left branch. More
precisely the set $\cal{L}$ of left contexts is defined by the
following grammar: $ \cal{L}$ $ := [] \ | \ \l  x. \cal{L} \ | \
(\cal{L} \ \cal{M})$. The result is thus a special case of the
following claim.

\noindent {\em Claim} : Let $L$ be a left context and $t$ be a
term. If $L[t]$ is in $SN$ then so is $w=L[(\lambda x. t \ x)]$.

\noindent {\em Proof} : By induction on $\langle type(t),
\eta(L[t])\rangle$. We show that every reduct of $w$ is in $SN$.
There are 4 possibilities for the reduced redex. If it is
 in $L$ or in $t$, the result follows
immediately from the IH. If it is the $(\lambda x. t \ x)$
substituted in the hole of $L$ the result is clear. The last
situation is when the redex is created by the substitution in the
hole of $L$. These cases are given below. Note that the {\em
assoc} and $\b$ rules can only be used either in $t$ or in $L$.

- $t= \l y.t_1$ and $w \tr_{\d} L[\l y. (\l x.t_1 \ x)]= L'[(\l
x.t_1 \ x)]$ where $L'=L[\l y. []]$. The result follows from the
IH applied to $L'$ and $t_1$ (since $t_1$ can be given a type less
than the one of $t$).

- $L=L'[([] \ v)]$ and  $w \tr_{\g} L'[(\l x.(t \ v) \ x)]$. The
result follows from the IH applied to $L'$ and $t_1=(t \ v)$
(since $t_1$ can be given a type less than the one of $t$).
\end{proof}

\begin{theorem}[typed]\label{main}
Let $t\in SN$ and $\s \in SN$ be a fair substitution. Then $\s(t)
\in SN$.
\end{theorem}
\begin{proof}
Formally, what we prove is the following. Let $U = \{(t,\sigma,A)\
| \ t \in SN$, $\sigma \in SN$ and $A$ is assignable to each
$\sigma(x)\}$. Then, for all $(t,\sigma,A) \in U$, $\sigma(t) \in
SN$. Theorem  follows since, if $\sigma$ is fair, $(t,\sigma,A)
\in U$ for some $A$ .

 We assume all the derivations
are normal (see the remark after definition \ref{def}). The proof
is by induction on $\langle size(A),\eta(t),size(t), \eta(\s,t),
size(\s,t) \rangle$. We will have to use the induction hypothesis
to some
 $(t',\s', A')$ for which we have to give type derivations  and to show that the 5-uplet has
 decreased. For the types (since the verification is fastidious but easy) we give some details only for
 one example (the first time in case 1.c below) and, for the others, we simply say \textquotedblleft $type(t_1)< type(t_2)$"
 (resp. \textquotedblleft$type(t_1) = type(t_2)$\textquotedblright)
  instead of saying  something as
 \textquotedblleft $t_1$ can be given a type less than (resp. equal to) $type(t_2)$".

 Note that this theorem will be only used with unary substitutions
but its proof needs the general case because, starting with a
unary substitution, it may happen that we have to use the
induction hypothesis with a non unary substitution. It will be the
case, for example, in 1.c below.

 Let $(t,\sigma,A) \in U$. If $t$ is an abstraction or a variable the result is trivial. Thus
assume $t=(H \ \overrightarrow{M})$ where $H$ is an abstraction or
a variable and $n=lg(\overrightarrow{M})\geq 1$. Let
$\overrightarrow{N}=\s(\overrightarrow{M})$.

\noindent {\em Claim} : Let $\overrightarrow{P}$ be a (strict)
initial or
a final sub-sequence of $\overrightarrow{N}$. Then $(z \ \overrightarrow{P}) \in SN$. \\
{\em Proof} : Let $ \overrightarrow{Q}$ be the sub-sequence of
$\overrightarrow{M}$ corresponding to $\overrightarrow{P}$.  Then
$(z \ \overrightarrow{P})=\tau(t')$ where $t'=(z \
\overrightarrow{Q})$ and $\tau$ is the same as $\s$ for the
variables in $ \overrightarrow{Q}$ and $z \not\in dom(\tau)$. The
result follows from the IH since $size(t') < size(t)$.
\hspace{3cm} $\Box$

\medskip

We use Lemma \ref{cs-sn} to show that $\s(t) \in SN$.

\begin{enumerate}
\item Assume $\s(t)$ is $\d$-head reducible. We have to show that
$D[\s(t)] \in SN$. There are 3 cases to consider.
\begin{enumerate}
  \item If $t$ was already $\d$-head reducible, then
  $D[\s(t)]=\s(D[t])$ and the result follows from the IH.
  \item If $H$ is a variable and $\s(H)=\l x. \l y.a$, then
  $D[\s(t)]=t'[z:=\l y. (\l x.a \ N(1))]$ where $t'=(z\ tail(\overrightarrow{N}))$.
  By the claim, $t'\in SN$ and since $type(z) <
  size(A)$ it is enough, by the IH,  to check that $\l y. (\l x.a \ N(1))
  \in SN$. But this is $\l y. (z' \ N(1))[z':=\l x. a]$. But, by the claim, $(z' \ N(1))\in SN$ and we conclude by the IH since
  $type(z') < size(A)$.
  \item If $H=\l x. z$ and $\s(z)=\l y. a$, then $D[\s(t)]=(\l y. (\l x. a \ N(1)) \ tail(\overrightarrow{N}))
  =\t(t')$ where $t'=(z' \ tail(\overrightarrow{M}))$ and $\t$ is the same as $\s$ on the variables of
 $tail(\overrightarrow{M})$ and $\t(z')=\l y.(\l x. a \ N(1))$.
 Note that, by Lemma \ref{prepa},   $t'$ is in
 $SN$ and $\eta(t') \leq \eta(t)$. Since $size(t') < size(t)$ to get the result by the IH we
  have to show that
(1)  $(t', \tau, A) \in U$ and (2)  that $(\l x. a \ N(1))\in SN$.

 To prove (1) it is enough to show that we can give to $Q=\l y.(\l x. a \ M(1))$ the same type as
 $P=(\l x. \l y.a \ M(1))$. In the typing of $P$, $\l x. \l y.a$ has type $(A_1 \f B_1 \f C_1) \et ... \et (A_k \f B_k \f C_k)$
 and $M(1)$ has type $A_1 \et  ... \et A_k$ and thus $P$ has type
 $( B_1 \f C_1) \et ... \et ( B_k \f C_k)$. It follows that we can type $Q$ by typing $(\l x. a \ M(1))$ with type
$C_1 \et  ... \et C_k$ and thus $Q$ with type $( B_1 \f C_1) \et
... \et ( B_k \f C_k)$.

 To prove (2) we remark that $(\l x. a \ N(1))=(\l x. z'' \
  N(1))[z'':=a]$ and, since $type(a) < size(A)$ it is enough, by the IH,  to
  show that $u=(\l x. z'' \ N(1)) \in SN$. This is done as follows: $u=\s'(t'')$ where $t''= (\l x. z'' \ M(1))$ (which is, up to
  the renaming of $z$ into $z''$  a sub-term of $t$) and $\s'$ is as $\s$ but where $z''$ is not in the domain of  $\s'$
  whereas the occurrence of $z$ in $H$ was in the domain of $\s$. Thus, $size(\s',t'')<size(\s,t)$ and the result follows from the IH.
\end{enumerate}

\item Assume $\s(t)$ is $\g$-head reducible. We have to show that
$L[\s(t)] \in SN$. There are 4 cases to consider.

\begin{enumerate}
  \item If $H$ is an abstraction, then $C[\s(t)]=\s(C[t])$ and the result follows immediately from
the IH.
\item  $H$ is a variable and $\s(H)=\l y. a$, then
$C[\s(t)]=(\l y.(a\ N(2))\ N(1) \ N(3)$ $... \ N(n))= (\l y.(a\
N(2))\ y \ N(3)$ $ ... \ N(n))[y:=N(1)]$. Since $type(N(1)) <
size(A)$, it is enough, by the IH, to show $(\l y.(a\ N(2))\ y \
N(3)$ $... \ N(n)) \in SN$ and so, by Lemma \ref{prepa2}, that
$u=(a\ N(2) \ N(3) \ ... \ N(n)) \in SN$.  By the claim,  $(z \
tail(\overrightarrow{N})) \in SN$ and the result follows from the
IH since $u=(z \ tail(\overrightarrow{N}))[z:=a]$ and $type(a) <
size(A)$.
\item $H$ is a variable and $\s(H)=(\l y. a \ b)$,  then
$C[\s(t)]=(\l y.(a\ N(1))\ b$ $N(2) \ ... \ N(n))=(z \
tail(\overrightarrow{N}))[z:= (\l y.(a\ N(1))\ b)]$. Since
$type(z) < size(A)$, by the IH it is enough to show that $u=(\l
y.(a\ N(1))\ b) \in SN$. We use Lemma \ref{cs-sn}.

- We first have to show that $B[u]\in SN$. But this is $(a[y:=b] \
N(1))$ which is in $SN$ since $u_1=(a[y:=b] \ \overrightarrow{N})
\in SN$ since $u_1=\t(t_1)$ where  $t_1$ is the same as $t$ but
where we have given to the variable $H$  the fresh name $z$, $\t$
is the same as $\s$ for the variables in $dom(\s)$ and
$\t(z)=a[y:=b]$ and thus we may conclude by the IH since $\eta(\t,
t ) < \eta(\s, t)$.

- We then have to show that, if $b$ is a $\b$-redex say $(\l z.
b_1 \ b_2 )$, then $A[u,1]=(\l z. (\l y. a \ N(1) \ b_1) \ b_2)
\in SN$. Let $u_2=\t(t_2)$ where $t_2$ is the same as $t$ but
where we have given to the variable $H$ the fresh name $z$, $\t$
is the same as $\s$ for the variables in $dom(\s)$ and
$\t(z)=A[\s(H),1]$. By the IH, $u_2 \in SN$.  Note that that $t_2
\in SN$, $\eta(t_2) \leq \eta(t)$ by Lemma \ref{prepa} and that
$\eta(\tau,t2)<\eta(\sigma,t)$. But $u_2=(\l z. (\l y.a \ b_1) \
b_2 \ \overrightarrow{N})$ and thus $u_3=(\l z. (\l y.a \ b_1) \
b_2 \ \ N(1)) \in SN$.  Since $u_3$ reduces to $A[u,1]$ by using
twice by the $\g$ rule, it follows that $A[u,1] \in SN$.
  \item If $H$ is a variable and $\s(H)$ is $\g$-head reducible, then
  $C[\s(t)]=\t(t')$ where $t'$ is the
same as $t$ but where we have given to the variable $H$  the fresh
name $z$ and $\t$ is the same as $\s$ for the variables in
$dom(\s)$ and $\t(z)=C[\s(H)]$. The result follows then from the
IH since $\eta(\tau,t')<\eta(\sigma,t)$.

\end{enumerate}

\item Assume that $\s(t)$ is $\b$-head reducible. We have to show that $Arg[\s(t)]\in SN$ and
 that $B[\s(t)] \in SN$. There are 3 cases to consider.

\begin{enumerate}
  \item If $H$ is an abstraction, the result follows immediately from
the IH since then $Arg[\s(t)]= \s(Arg[t])$ and
$B[\s(t)]=\s(B[t])$.

  \item If $H$ is a variable  and $\s(H)=\l y. v$
for some $v$. Then $Arg[\s(t)]= N(1) \in SN $ by the IH and
$B[\s(t)]=(v[y:=N(1)] \ tail(\overrightarrow{N}))=(z \
tail(\overrightarrow{N}))[z:=v[y:=N(1)]]$. By the claim, $(z \
tail(\overrightarrow{N})) \in SN$. By the IH, $v[y:=N(1)] \in SN$
since  $type(N(1)) < size(A)$.  Finally the IH implies that
$B[\s(t)] \in SN$ since $type(v) < size(A)$.

  \item $H$ is a variable  and $\s(H)=(R
\ \overrightarrow{M'})$ where $R$ is a $\b$-redex. Then
$Arg[\s(t)]=Arg[\s(H)] \in SN$ and $B[\s(t)]=(R' \
\overrightarrow{M'} \ \overrightarrow{N})$ where $R'$ is the
reduct of $R$. But then $B[\s(t)] =\t(t')$  and $t'$ is the same
as $t$ but where we have given to the variable $H$  the fresh name
$z$ and $\t$ is the same as $\s$ for the variables in $dom(\s)$
and $\t(z)=(R' \ \overrightarrow{M'})$. Note that  that $t' \in
SN$ and $\eta(t')\leq \eta(t)$, by Lemma 3.1. We conclude by the
IH since $\eta(\t,t') < \eta(\s,t)$.
\end{enumerate}

\item We, finally, have to show that, for each $i$, $A[\s(t), i] \in
SN$. There are again 3 cases to consider.

\begin{enumerate}

  \item If the $\b$-redex put in head position is some $N(j)$ and
  $M(j)$ was already a redex. Then $A[\s(t), j]=\s(A[t,j])$ and
  the result follows from the IH.

  \item If the $\b$-redex put in head position is some $N(j)$ and
  $M(j)=(x \ a)$ and $\s(x)=\l y. b$ then $A[\s(t), i]= \l y. (\s(H) \
N(1) \ ... \ N(j-1) \ b) \ \s(a) \ N(j+1) \ ... \ N(n))$. Since
$type(\s(a)) < size(A)$ it is enough, by the IH, to show that $\l
y. (\s(H) \ N(1) \ ... \ N(j-1) \ b) \ y \ N(j+1) \ ... \ N(n))$
and so, by Lemma \ref{prepa2}, that $(\s(H) \ N(1) \ ... \ N(j-1)
\ b \ N(j+1) \ ... \ N(n)) \in SN$. Since $type(b) < size(A)$ it
is enough, by the IH, to show $u=(\s(H) \ N(1) \ ... \ N(j-1) \ z
\ N(j+1) \ ... \ N(n)) \in SN$. Let $t'= (H \
\overrightarrow{M'})$ where $\overrightarrow{M'}$ is defined by
$M'(k)=M(k)$, for $k \neq j$, $M'(j)=z$.  Since $t= t'[z:=(x \
a)]$ and $u=\s(t')$ the result follows from Lemma \ref{prepa} and
the IH.

\item If,  finally, $H$ is a variable, $\s(H)=(H' \
\overrightarrow{M'})$ and the $\b$-redex put in head position is
some $M'(j)$. Then, $A[\s(t),j]=\t(A[t',j])$ where $t'$ is the
same as $t$ but where we have given to the variable $H$ the fresh
variable $z$ and $\t$ is the same as $\s$ for the variables in
$dom(\s)$ and $\t(z)=A[\s(H),j]$.  Note that  that $t' \in SN$ and
$\eta(t')\leq \eta(t)$, by Lemma 3.1. We conclude by the IH since
$\eta(\t,t') < \eta(\s,t)$.
\end{enumerate}
\end{enumerate}
\end{proof}

\begin{corollary}\label{cor}
Let $t$ be a typable  term. Then $t$ is strongly normalizing.
\end{corollary}
\begin{proof}
By induction on $size(t).$ If $t$ is an abstraction or a variable
the result is trivial.  Otherwise $t=(u\ v)$ and, by the IH, $u,v
\in SN$. Thus, by Theorem \ref{main}, $(u \ y)=(x\ y)[x:=u] \in
SN$ and, by applying again Theorem \ref{main}, $(u \ v)= (u \
y)[y:=v] \in SN$.
\end{proof}

\end{document}